\documentclass{llncs}

\newcommand{\ignore}[1]{}

\newcommand{\dom}{\textit{Dom}} 

\newcommand\inter[1]{\llbracket #1 \rrbracket}

\newcommand{\sst}{\textnormal{SST}} 
\newcommand{\seq}[1]{\langle #1 \rangle}

\newcommand{\Vars}{\mathcal{X}}

 \usepackage{amsmath}
 \usepackage{amssymb}
 \usepackage{stmaryrd}
 \usepackage{xspace}
 \usepackage{placeins}
 \usepackage{enumerate}
\usepackage{subfigure}
\usepackage{verbatim}
\usepackage{enumerate}
\usepackage{algorithm}
\usepackage[latin1]{inputenc}

\usepackage{xcolor,pgf,tikz,pgflibraryarrows,pgffor,pgflibrarysnakes}
\usetikzlibrary{fit} 
\usetikzlibrary{backgrounds} 
\usepgflibrary{shapes}
\usetikzlibrary{snakes,automata}

\tikzstyle{background}=[rectangle,fill=gray!10, inner sep=0.1cm, rounded corners=0mm]

\begin{document}
\title{On Streaming String Transducers and HDT0L Systems}
\author{Emmanuel Filiot~\inst{1} \and Pierre-Alain Reynier~\inst{2}}
\institute{Computer Science department, Universit\'e Libre de Bruxelles, Belgium \and Aix-Marseille Universit\'e, LIF, France}

\maketitle

\begin{abstract}
    Copyless streaming string transducers (copyless SST) have been
    introduced by R.~Alur and
    P.~Cerny in 2010 as a one-way deterministic automata model to define
    transformations of finite strings. Copyless SST extend deterministic finite state automata
    with a set of registers in which to store intermediate output
    strings, and those registers can be combined and updated all along
    the run, in a linear manner, i.e., no register content can be
    copied on transitions. It is known that copyless SST capture exactly the class of
    MSO-definable string-to-string transformations, as defined by B.~Courcelle, and are
    equi-expressive to deterministic two-way transducers. They enjoy
    good algorithmic properties. Most notably, they have decidable
    equivalence problem (in \textsf{PSpace}). In this paper, we show
    that they still have decidable equivalence problem even without
    the copyless restriction. The proof reduces to the HDT0L sequence
    equivalence problem, which is known to be decidable. We also show
    that this latter problem is as difficult as the SST equivalence
    problem, modulo linear time reduction. 
\end{abstract}

\section{Introduction}

The theory of languages is extremely rich and important automata-logic
correspondances have been shown for various classes of logics, automata,
and structures. There are less known automata-logic connections in the
theory of transformations. Nevertheless, important results have been
obtained for the class of MSO-definable transformations, as defined by
Courcelle. Most notably, it has been shown by J.~Engelfriet and
H.J.~Hoogeboom that MSO-definable (finite) string to string
transformations are exactly those transformations defined by 
deterministic two-way transducers~\cite{EH01}. This result has then been extended to 
ordered ranked trees by J.~Engelfriet and S.~Maneth, for the class of
linear-size increase macro tree transducers~\cite{MTT}. More recently,
MSO-definable transformations of finite strings have also been
characterized by a new automata model, that of (copyless) streaming
string transducers, by R.~Alur and P.~Cerny~\cite{AC10}.

Copyless streaming string transducers (SST) extend deterministic finite state
automata with a finite set of string variables $X,Y,\dots$. Each variable stores an
intermediate string output and can be combined and updated with other 
variables. Along the transitions, a finite string can be appended or
prepended to a variable, and variables can be reset or
concatenated. The variable updates along the transitions are formally
defined by variable substitutions and the copyless restriction
is defined by considering only linear substitutions. Therefore,
variable update such as $X := XX$ are forbidden by the copyless
restriction. The SST model has then been extended to other structures
such as infinite strings~\cite{DBLP:conf/lics/AlurFT12}, 
trees~\cite{DBLP:conf/icalp/AlurD12}, and quantitative
languages~\cite{DBLP:conf/lics/AlurDDRY13}.

Copyless SST enjoy many good algorithmic properties such as decidable
equivalence problem (given two copyless SST, do they define the same 
transformation?) or decidable typechecking (given a copyless SST $T$
and two finite automata $A,B$, does $T(L(A))\subseteq L(B)$
hold?). While the copyless SST equivalence problem is known to be
decidable in \textsf{PSpace}, it is unknown whether decidability still
holds without the copyless restriction. In this paper, we close this
question and show that (copyfull) SST have decidable equivalence
problem. We reduce this problem to the HDT0L sequence equivalence
problem, which is decidable~\cite{DBLP:journals/tcs/CulikK86}, but its complexity is
unknown. We actually prove a slightly more general result on the class
of non-deterministic SST: we show that checking whether a
given non-deterministic SST defines a function is a decidable problem, again by
reduction to the HDT0L equivalence problem. Conversely, we show that the SST equivalence problem is as
hard as the HDT0L sequence equivalence problem, modulo linear time
reduction.

\section{Streaming String Transducers}

For all finite alphabets $\Sigma$, we denote by $\Sigma^*$ the set of
finite words over $\Sigma$, and by $\epsilon$ the empty word. 
Let $\Vars$ be a finite set of variables denoted by $X,Y,\dots$ and $\Gamma$ be a finite alphabet. 
A substitution $\sigma$ is defined as a mapping ${\sigma : \Vars \to (\Gamma \cup
  \Vars)^*}$. Let $\mathcal{S}_{\Vars, \Gamma}$ be the set of all substitutions.
Any substitution $\sigma$ can be extended to $\hat{\sigma}: (\Gamma \cup \Vars)^*
\to (\Gamma \cup \Vars)^*$ in a straightforward manner.
The composition $\sigma_1 \sigma_2$  of two substitutions $\sigma_1$ and $\sigma_2$
is defined as the standard function composition  $\hat{\sigma_1} \sigma_2$,
i.e. $\hat{\sigma_1}\sigma_2(X) = \hat{\sigma_1}(\sigma_2(X))$ for all $X \in \Vars$. 
\begin{definition}
  \label{def:sst}
  A streaming string transducer (\sst for short) is a tuple 
  $T = (\Sigma, \Gamma, Q, q_0, Q_f, \delta, \Vars, \rho, F)$ where:
  \begin{itemize}
  \item
    $\Sigma$ and $\Gamma$ are finite sets of input and output alphabets;
  \item 
    $Q$ is a finite set of states with initial state $q_0$;
  \item 
    $\delta : Q \times \Sigma \to Q$ is a transition function;
  \item 
    $\Vars$ is a finite set of variables;
  \item 
    $\rho : \delta \to \mathcal{S}_{\Vars, \Gamma}$ is a variable update
    function ;  
  \item $Q_f$ is a subset of final states;
  \item 
    $F: Q_f \rightharpoonup \Vars^*$ is an output  function.
  \end{itemize}
\end{definition}

The concept of a run of an \sst{} is defined in an analogous manner to that of
a finite state automaton.
The sequence $\seq{\sigma_{r, i}}_{0 \leq i \leq |r|}$ of substitutions induced
by a run $r = q_0 \xrightarrow{a_1} q_1 \xrightarrow{a_2} q_2 \ldots q_{n-1}\xrightarrow{a_n} q_n$ is defined
inductively as the following: $\sigma_{r, i} {=} \sigma_{r, i{-}1} \rho(q_{i-1}, a_{i})$
for $1 < i \leq |r|$ and $\sigma_{r,1} = \rho(q_0,a_1)$. We denote $\sigma_{r,|r|}$ by $\sigma_r$.

If $r$ is accepting, i.e. $q_n\in Q_f$, we can extend the output function $F$ to $r$ by 
$F(r) = \sigma_\epsilon\sigma_{r}F(q_n)$, where $\sigma_\epsilon$
substitute all variables by their initial value $\epsilon$. For all
words $w\in\Sigma^*$, the output of $w$ by $T$ is defined only if there exists an accepting run $r$ of $T$ on $w$, and in that case the output
is denoted by $T(w) = F(r)$. The \emph{domain} of $T$, denoted by
$\dom(T)$, is defined as the set of words $w$ on which there exists
an accepting run of $T$. The transformation $\inter{T}$ defined by
$T$ is the function which maps any word $w\in \dom(T)$ to its output
$T(w)$.

In~\cite{AC10}, the variable update are required to be
\emph{copyless}, i.e. no variable can occur more than once in the rhs of the
substitutions $\rho(t)$, for all $t\in \delta$. One of the main result
of~\cite{AC10} is to show that this restriction allows one to capture
exactly the class of MSO-definable transformations. 

\paragraph{BiSST} A \emph{biSST} $T = (\Sigma, \Gamma, Q, q_0, Q_f,
\delta, \Vars, \rho, F)$ is defined as an SST, except for the output
function $F$, which is a function from $Q_f$ to $\Vars^*\times
\Vars^*$. The semantics of SST can be naturally extended to biSST such
that $\inter{T}$ is a function from $\Sigma^*$ to $\Gamma^*\times
\Gamma^*$. A biSST $T$ is \emph{diagonal} if for all $w\in\Sigma^*$,
$T(w) = (s,s)$ for some $s\in\Gamma^*$.

\paragraph{Non-deterministic SST and biSST} SST and biSST can be
naturally extended with non-determinism. In the non-deterministic
setting, transitions are defined by a relation $\delta\subseteq
Q\times \Sigma \times Q$. Since there might be several accepting runs
associated with an input word $w$, a non-deterministic SST
(resp. biSST) $T$ defines a relation from $\Sigma^*$ to $\Gamma^*$,
(resp. from $\Sigma^*$ to $\Gamma^*{\times}\Gamma^*$), defined as the 
set of pairs $(w,F(r))$ such that $r$ is an accepting run of $T$ on
$w$. A non-deterministic SST $T$ is \emph{functional} if it defines a
function, i.e. for all words $w\in\Sigma^*$, $|\{ w'\ |\
(w,w')\in\inter{T}\}|\leq 1$. A non-deterministic biSST $T$ is
\emph{diagonal} if for all words $w,w_1,w_2\in \Sigma^*$, if 
$(w,(w_1,w_2))\in\inter{T}$, then $w_1=w_2$.

\paragraph{Synchronised Product of SST} Let $T_1$ and $T_2$  be two non-deterministic
SST: $T_1 {=} (\Sigma, \Gamma, Q^{(1)},  q_0^{(1)}, \allowbreak Q_f^{(1)},
\delta^{(1)}, \Vars^{(1)}, \rho^{(1)}, F^{(1)})$ and $T_2 = (\Sigma, \Gamma, Q^{(2)}, q_0^{(2)}, Q_f^{(2)},
\delta^{(2)},\allowbreak \Vars^{(2)}, \rho^{(2)}, F^{(2)})$. For all
transitions $t_1 = (q_1,a_1,p_1)\in \delta^{(1)}$ and $t_2 =
(q_2,a_2,p_2)\in \delta^{(2)}$, the product transition $t_1\otimes
t_2$ is defined only if $a_1=a_2$, by $((q_1,q_2),a_1,(p_1,p_2))$. 
We define the product $T_1\otimes T_2$ of $T_1$
and $T_2$ as a biSST that simulates both SST $T_1$ and $T_2$ in
parallel on the same inputs, and produces the pair of respective outputs
of these SSTs. Formally, $T_1\otimes T_2$ is the biSST
$T = (\Sigma, \Gamma, Q^{(1)}\times Q^{(2)}, (q_0^{(1)},q_0^{(2)}),
Q_f^{(1)}\times Q_f^{(2)}, \delta^{(1)}\otimes \delta^{(2)}, \Vars^{(1)}\uplus \Vars^{(2)},
\rho, F)$ where:
\begin{itemize}
    \item for all $t_1\otimes t_2\in\delta^{(1)}\otimes \delta^{(2)}$, all $X\in \Vars^{(1)}\uplus
      \Vars^{(2)}$, $\rho(t)(X) = \rho^{(i)}(t_i)(X)$ if $X\in
      \Vars^{(i)}$
    \item $F$ is defined only on $Q_f^{(1)}\times Q_f^{(2)}$, by
      $F(q_1,q_2) = (F^{(1)}(q_1),F^{(2)}(q_2))$ for all $(q_1,q_2)\in
      Q_f^{(1)}\times Q_f^{(2)}$.
\end{itemize}

The following proposition is immediate by construction of $T_1\otimes T_2$:
\begin{proposition}
Given two non-deterministic SST $T_1$ and $T_2$, the following holds:
\begin{enumerate}
\item $\dom(T_1\otimes T_2) = \dom(T_1)\cap \dom(T_2)$, 
\item $\inter{T_1\otimes T_2} = \{ (w,(w_1,w_2))\ |\ (w,w_1)\in
\inter{T_1}\wedge (w,w_2)\in \inter{T_2}\}$, 
\item if both $T_1$ and $T_2$ are deterministic, then so is
  $T_1\otimes T_2$. 
\end{enumerate}
\end{proposition}

\section{SST Functionality and Equivalence Problems}
\label{sec:prelims}

The functionality problem for non-deterministic copyless SSTs, as well as
the equivalence problem of (deterministic) copyless SSTs, are both known to be
decidable in \textsf{PSpace}~\cite{conf/icalp/AlurD11}. We show that both problems are
still decidable even if one drops the copyless assumption.

More precisely, the \emph{functionality} problem asks, given a non-deterministic SST $T$, whether it is
functional, i.e., whether $\inter{T}$ is a function. The
\emph{equivalence} problem asks, given two (deterministic) SST $T_1$
and $T_2$, whether they define the same function, i.e. $\inter{T_1} =
\inter{T_2}$. We show that both problems reduce to the diagonal
problem of biSST: given a biSST $T$, it amounts to decide whether $T$
is diagonal. The following result is indeed trivial:

\begin{proposition}
    Let $T$ be a non-deterministic SST. Then $T$ is functional iff 
    $T\otimes T$ is diagonal. 
\end{proposition}

The equivalence problem of (deterministic) SST can be reduced to the
functionality problem of non-deterministic SST:
 
\begin{proposition}
    Let $T_1,T_2$ be two deterministic SST. Then $T_1$ and $T_2$
    are equivalent iff 
    \begin{enumerate}
        \item $\dom(T_1) = \dom(T_2)$, and
        \item $T_1\otimes T_2$ is diagonal.
    \end{enumerate}
\end{proposition}

In the rest of this section, we therefore focus on the diagonal
problem for non-deterministic biSST, and show its decidability by reduction to the HDT0L
equivalence problem.

\paragraph{HDT0L Sequence Equivalence Problem} 
An \emph{HDT0L instance} $\mathcal{I}$ is given two finite alphabet $A,B$, $2n+2$ morphisms $h_1,\dots,h_n,h,g_1,\dots,g_n,g$ such that
$h_i,g_i:A^*\rightarrow A^*$ and $h,g:A^*\rightarrow B^*$, and two words $v,w\in A^*$. 
The \emph{HDT0L sequence equivalence problem} asks
to decide, given such an instance $\mathcal{I}$, whether for all  sequences $i_1,\dots,i_k$ such that $i_j\in
\{1,\dots,n\}$, the following holds
$$
h(h_{i_1}(\dots h_{i_k}(v)\dots)) = 
g(g_{i_1}(\dots g_{i_k}(w)\dots))
$$
In this case, we say that $\mathcal{I}$ is valid. 
This problem is known to be decidable~\cite{DBLP:journals/tcs/CulikK86} but its precise
complexity unknown.

\begin{lemma}
    Given a biSST $T$, one can construct in polynomial time an
    HDT0L instance $\mathcal{I}_T$ such that $T$ is diagonal iff 
    $\mathcal{I}_T$ is valid. 
\end{lemma}

\begin{proof}
Let $T = (\Sigma, \Gamma, Q, q_0, Q_f, \delta, \Vars, \rho, F)$. 
We define the two alphabets of $\mathcal{I}_T$ as follows:
\begin{itemize}
\item $A= \{ \alpha_q \mid \alpha\in\mathcal{X}\cup \Gamma \cup \{\$,\#\} , q \in Q\} \cup \{ \#\}$
\item $B=\Gamma$
\end{itemize}

We first introduce the morphism 
$\text{subscript}_q: (\mathcal{X}\cup \Gamma \cup \{\$,\#\})^* \rightarrow A^*$, for every state $q\in Q$,
and $\alpha\in \mathcal{X}\cup \Gamma \cup \{\$,\#\}$, by
$\text{subscript}_q(\alpha)=\alpha_q$.

We then define the following morphisms:
\begin{itemize}
    \item Morphism $f_q^i:A^* \rightarrow A^*$, where $q\in Q_f$ and $i=1,2$:
    $$f_q^i (\alpha) = \left \{
    \begin{array}{ll}
     \text{subscript}_q(\$\pi_i(F(q))) & \text{ if } \alpha= \# \\
     \epsilon                                     & \text{ otherwise.}        
    \end{array}
    \right .        
        $$

    \item Morphism $f_t:A^* \rightarrow A^*$, where $t = (q,a,q')\in\delta$:
    $$f_t(\alpha) = \left \{
    \begin{array}{ll}
     \text{subscript}_q(\rho(t)(X))        & \text{ if }\alpha = X_{q'} \text{ with }X\in \mathcal{X} \\
     \text{subscript}_q(a)                   & \text{ if }\alpha = a_{q'} \text{ with }a\in \Gamma \cup \{\$\} \\
     \epsilon                                        & \text{ otherwise }       
    \end{array}
    \right .        
        $$
    \item Morphism $f_{q_0}:A^* \rightarrow B^*$:
    $$f_{q_0}(\alpha) = \left \{
    \begin{array}{ll}
     a              & \text{ if } \alpha = a_{q_0} \text{ with }a\in \Gamma   \\
     \epsilon    & \text{ otherwise}      
    \end{array}
    \right .        
        $$
%
\end{itemize}

We define the instance $\mathcal{I}_T$ of the
HDT0L equivalence problem by:

\begin{itemize}
    \item $v=w=\#$
    \item the pairs $(h_q,g_q) = (f_q^1,f_q^2)$ for all $q\in Q_f$
    \item the pairs $(h_t,g_t) = (f_t,f_t)$ for all transitions $t\in\delta$
    \item the final pair $(h,g) = (f_{q_0}, f_{q_0})$
\end{itemize}

We claim that $T$ is diagonal iff $\mathcal{I}_T$ is valid.

Consider a sequence $i_1,\ldots,i_k$ such that $i_j\in Q_f \cup \delta$.
We say it is \emph{valid} if $i_j\in Q_f$ iff $j=k$, and the sequence of transitions $i_1,\ldots i_{k-1}$
is a sequence of consecutive transitions of $T$, and verifies that the source of $i_1$ is the initial state $q_0$ of $T$, 
and that the target state of $i_{k-1}$ is the state $i_k$.

We state the two following properties which can be proven by induction :
\begin{enumerate}[$(i)$]
\item given an accepting run $r : q_0 \xrightarrow{u} q$ of $T$, corresponding to transitions 
$t_1,\ldots,t_n$, the following equality holds for $i=1,2$:
$$f_{q_0}(f_{t_1}(\dots f_{t_n}(f_q^i(\#)\dots))) = \pi_i(T(u))$$
\item given any sequence $i_1,\ldots,i_k$ which is not valid, we have 
$$h(h_{i_1}(\dots h_{i_k}(\#)\dots)) = \epsilon =  g(g_{i_1}(\dots g_{i_k}(\#)\dots))$$
\end{enumerate}

First suppose that $\mathcal{I}_T$ is valid, we show that $T$ is diagonal. Take
$u=a_1\dots a_n\in \dom(T)$ and $t_1\dots t_n$ the sequence of
transitions of $T$ on $u$. Suppose that $q$ is the final state reached
by $t_n$. Then, we have by hypothesis 
$$
f_{q_0}(f_{t_1}(\dots f_{t_n}(f_q^1(\#)\dots))) = 
f_{q_0}(f_{t_1}(\dots f_{t_n}(f_q^2(\#)\dots)))
$$
from which we get $\pi_1(T(u)) = \pi_2(T(u))$.

Conversely, either the sequence $i_1,\ldots,i_k$ is not valid, and then equality follows from the  
property $(ii)$ stated above, or it is valid, and it then follows from property $(i)$ and from the fact that $T$
is diagonal.
\end{proof}

\section{From HDT0L sequence equivalence to SST equivalence problem}

In this section, we show that the SST equivalence problem is as
difficult as the HDT0L sequence equivalence problem.  

\begin{lemma}
    Given an HDT0L instance $\mathcal{I}$, one can construct in
    linear time two (deterministic) SST $T_1$ and $T_2$ such that $T_1$ and $T_2$ are
    equivalent iff $\mathcal{I}$ is valid. 
\end{lemma}

\begin{proof}
Let $\mathcal{I} = (A,B, g,h,g_1,h_1,\dots,g_n,h_n,v,w)$ be an HDT0L instance. 
We construct two SSTs $T_1$ and $T_2$ such that $\mathcal{I}$ is valid
iff $T_1$ and $T_2$ are equivalent. The two SST both have two states $q_0,q_1$ and
$p_0,p_1$ respectively, where $q_0,p_0$ are initial, and $p_1,q_1$ are
final. The input alphabet of both $T_i$ is $\mathbb{N}_n = \{0,1,\dots,n\}$ and the
output alphabet is $B$. They both have the same set of variables
$\Vars = \{ X_a\ |\ a\in A\}$. Their transitions are defined by:

\begin{itemize}

  \item for $T_1$: $q_0\xrightarrow{0} q_1$, and for all
    $i\in\mathbb{N}_n\setminus \{0\}$, $q_1\xrightarrow{i} q_1$,
  \item for $T_2$: $p_0\xrightarrow{0} p_1$, and for all $i\in\mathbb{N}_n\setminus \{0\}$, $p_1\xrightarrow{i} p_1$.
\end{itemize}

To define the update functions, we first introduce the morphism
$\text{rename}_X : A^*\rightarrow \Vars^*$ defined for all $a\in A$ by
$\text{rename}_X(a) = X_a$. Then, the update functions $\rho_1$ and $\rho_2$ of $T_1$ and $T_2$
respectively are defined, for all $a\in A$, by:

\begin{itemize}
  \item $\rho_1(q_0,0,q_1)(X_a) = h(a)$,
  \item $\rho_1(q_1,i,q_1)(X_a) = \text{rename}_X(h_i(a))$ for all
    $i\in \mathbb{N}_n\setminus \{0\}$,
  \item $\rho_2(p_0,0,p_1)(X_a) = g(a)$,
  \item $\rho_2(p_1,i,p_1)(X_a) = \text{rename}_X(g_i(a))$ for all
    $i\in \mathbb{N}_n\setminus \{0\}$.
\end{itemize}

Finally, the respective output functions of $T_1$ and $T_2$ are 
defined by $F_1(q_1) = \text{rename}_X(v)$ and $F_2(p_1) = \text{rename}_X(w)$.

In order to show the correctness of the construction, it suffices to
remark that given a run $r_1 = q_0\xrightarrow{0}
q_1\xrightarrow{i_1}q_1\dots q_1\xrightarrow{i_k} q_1$ in $T_1$, the
output $F_1(r_1)$ of $r_1$ is exactly $h(h_{i_1}(\dots
h_{i_k}(v)\dots))$, i.e. $T_1(0i_1\dots i_k) =  h(h_{i_1}(\dots
h_{i_k}(v)\dots))$. Similarly, $T_2(0i_1\dots i_k) = g(g_{i_1}(\dots
g_{i_k}(w)\dots))$.

\begin{example} Let us illustrate this construction on some example. Consider the
    following HDT0L instance $(A,B,h,g,h_1,g_1,v,w)$ with $A = 
    \{a,b,c,d\}$, $B = \{ e,f\}$, $v = c$ and $w = cd$, and the
    morphisms are defined as follows:
    $$
    \begin{array}{cccccccccccccccccccccccccc}

        h = g & : & a & \rightarrow & e &\qquad h_1 & : & a & \rightarrow & a &
        \qquad g_1 & : &  a & \rightarrow a \\
           &   &  b & \rightarrow & f & & & b & \rightarrow & b & & &
           b & \rightarrow b \\
            &   &  c & \rightarrow & \epsilon & & & c & \rightarrow &
            acb & & & c & \rightarrow & ca\\
            &   & d & \rightarrow & \epsilon & & &  d & \rightarrow &
            \epsilon & & & d & \rightarrow & db 
    \end{array}
    $$

    For instance, we can obtain the following deriviations:
    $$
    \begin{array}{ccccccccccccccccccccccc}
    v = c & \xrightarrow{h_1} & acb
    & \xrightarrow{h_1} & a^2cb^2 & \xrightarrow{h_1} & a^3cb^3
    & \xrightarrow{h} & e^3f^3 \\

    w = cd & \xrightarrow{g_1} & cadb
    & \xrightarrow{g_1} & ca^2db^2 & \xrightarrow{g_1} & ca^3db^3
    & \xrightarrow{g} & e^3f^3 \\
    \end{array}
    $$
The two constructed SSTs are then the following:

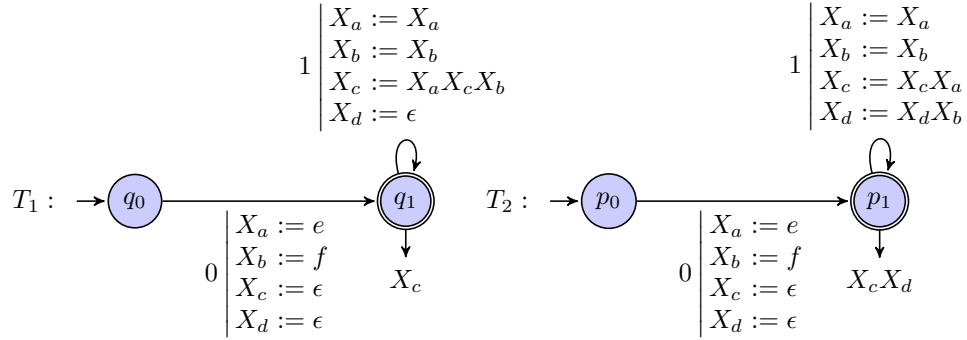
\begin{center}
\begin{tikzpicture}[->,>=stealth',shorten >=1pt,auto,node distance=1cm,
 semithick,scale=0.9]
  \tikzstyle{every state}=[fill=blue!20!white,minimum size=2em]
\tikzstyle{accepting-output}=[accepting by arrow, accepting
  where=below, accepting text=$X_c$]
 \node[initial, initial text={},state,fill=blue!20!white] at (1,0) (A) {$q_0$} ;
 \node[state,accepting-output, accepting, fill=blue!20!white] at (5,0) (B) {$q_1$} ;
 \path(A) edge  node[anchor=south,below]{$0 \left| \begin{array}{lll}
           X_a & := & 
       e \\ X_b & :=&  f \\ X_c& :=& \epsilon \\ X_d& :=& \epsilon\end{array}\right. $} (B);
  \path(B) edge [loop above] node[anchor=north,above]  {$1
    \left| \begin{array}{lll} X_a & := &
       X_a \\ X_b & :=&  X_b \\ X_c  & :=& X_aX_cX_b \\ X_d& :=& \epsilon\end{array}\right. $} (B);
 \node[draw=none] at (-0.5,0) {$T_1:$};

  \tikzstyle{every state}=[fill=blue!20!white,minimum size=2em]
\tikzstyle{accepting-output}=[accepting by arrow, accepting
  where=below, accepting text=$X_cX_d$]
 \node[initial, initial text={},state,fill=blue!20!white] at (8,0) (A) {$p_0$} ;
 \node[state,accepting-output, accepting, fill=blue!20!white] at (12,0) (B) {$p_1$} ;
 \path(A) edge  node[anchor=south,below]{$0 \left| \begin{array}{lll}
           X_a & := &
       e \\ X_b & :=&  f \\ X_c& :=& \epsilon \\ X_d& :=& \epsilon\end{array}\right. $} (B);
  \path(B) edge [loop above] node[anchor=north,above]  {$1
    \left| \begin{array}{lll} X_a & := & 
       X_a \\ X_b & := & X_b \\ X_c& :=& X_cX_a \\ X_d& :=& X_dX_b\end{array}\right. $} (B);
 \node[draw=none] at (6.5,0) {$T_2:$};
   \end{tikzpicture}
\end{center}
\end{example}

\end{proof}


\bibliographystyle{abbrv}

\begin{thebibliography}{1}

\bibitem{AC10}
R.~Alur and P.~{\v C}ern\'y.
\newblock {Expressiveness of streaming string transducers}.
\newblock In {\em FSTTCS}, volume~8, pages 1--12, 2010.

\bibitem{DBLP:conf/icalp/AlurD12}
R.~Alur and L.~D'Antoni.
\newblock Streaming tree transducers.
\newblock In {\em ICALP (2)}, volume 7392 of {\em LNCS}, pages 42--53.
  Springer, 2012.

\bibitem{DBLP:conf/lics/AlurDDRY13}
R.~Alur, L.~D'Antoni, J.~V. Deshmukh, M.~Raghothaman, and Y.~Yuan.
\newblock Regular functions and cost register automata.
\newblock In {\em 28th Annual {ACM/IEEE} Symposium on Logic in Computer
  Science, {LICS} 2013, New Orleans, LA, USA, June 25-28, 2013}, pages 13--22.
  {IEEE} Computer Society, 2013.

\bibitem{conf/icalp/AlurD11}
R.~Alur and J.~V. Deshmukh.
\newblock Nondeterministic streaming string transducers.
\newblock In {\em ICALP}, volume 6756 of {\em LNCS}, pages 1--20. Springer,
  2011.

\bibitem{DBLP:conf/lics/AlurFT12}
R.~Alur, E.~Filiot, and A.~Trivedi.
\newblock Regular transformations of infinite strings.
\newblock In {\em LICS}, pages 65--74. IEEE, 2012.

\bibitem{EH01}
J.~Engelfriet and H.~J. Hoogeboom.
\newblock {MSO} definable string transductions and two-way finite-state
  transducers.
\newblock {\em ACM Trans. Comput. Logic}, 2:216--254, 2001.

\bibitem{MTT}
J.~Engelfriet and S.~Maneth.
\newblock Macro tree transducers, attribute grammars, and mso definable tree
  translations.
\newblock {\em Information and Computation}, 154(1):34--91, 1999.

\bibitem{DBLP:journals/tcs/CulikK86}
K.~C. II and J.~Karhum{\"{a}}ki.
\newblock The equivalence of finite valued transducers (on {HDT0L} languages)
  is decidable.
\newblock {\em Theor. Comput. Sci.}, 47(3):71--84, 1986.

\end{thebibliography}

\end{document}